\documentclass[showpacs, amsmath, amssymb, amsthm, aps, jmp, thmsb, pra, twocolumn]{revtex4}

\usepackage{color}
\usepackage{float}
\usepackage{amsmath}
\usepackage{amssymb}
\usepackage{amsfonts}
\usepackage{amsbsy}
\usepackage{latexsym}
\usepackage{mathrsfs}
\usepackage{bbm}
\usepackage{bm}
\usepackage{graphicx}
\usepackage{epsfig}

\newtheorem{lemma}{Lemma}

\newenvironment{proof}[1][Proof]{\textbf{#1.} }{\ \rule{0.5em}{0.5em}}

\newcommand{\be}{\begin{eqnarray}}
\newcommand{\ee}{\end{eqnarray}}
\def\({\left(}
\def\){\right)}
\def\[{\left[}
\def\]{\right]}
\def\C{\mathbb{C}}

\newcommand{\braket}[1]{\left\langle #1 \right\rangle}

\newcommand{\bra}[1]{\langle #1 |}
\newcommand{\ket}[1]{| #1 \rangle}

\newcommand{\Tr}{\mathrm{Tr}}
\newcommand{\sla}[1]{\rlap{\kern .15em /}#1}

\newcommand{\ot}{\otimes}
\newcommand{\Bot}{\bigotimes}

\begin{document}
\title{Mixed State Entanglement Measures for Intermediate Separability}

\author{Tsubasa Ichikawa}
\affiliation{ Research Center for Quantum Computing, Interdisciplinary Graduate School of Science and Engineering,
 Kinki University, 3-4-1 Kowakae,
Higashi-Osaka, Osaka 577-8502, Japan}

\author{Marcus Huber}
\affiliation{Faculty of Physics, University of Vienna, A-1090 Vienna, Austria}

\author{Philipp Krammer}
\affiliation{Faculty of Physics, University of Vienna, A-1090 Vienna, Austria}

\author{Beatrix C. Hiesmayr}
\affiliation{Faculty of Physics, University of Vienna, A-1090 Vienna, Austria}

\begin{abstract}
To determine whether a given multipartite quantum state is separable with respect to some partition we construct a family of entanglement measures $\{R_m(\rho)\}$. This is done utilizing generalized concurrences as building blocks which are defined by flipping of $M$ constituents and indicate states that are separable with regard to bipartitions when vanishing. Further, we provide an analytically computable lower bound for  $\{R_m(\rho)\}$ via a simple ordering relation of the convex roof extension. Using the derived lower bound, we illustrate the effect of the isotropic noise on a family of four-qubit mixed states for each intermediate separability.
\end{abstract}

\maketitle

\section{Introduction}

Quantum entanglement plays a crucial role in foundations of quantum physics and is an indispensable ingredient for quantum information processing tasks \cite{nielsen00}. In recent years it has become important to quantify the entanglement of quantum states, since not all entangled states are equally useful for quantum protocols. In particular, for multipartite systems it is of high interest to detect and quantify the entanglement not only of the whole system, but also between various constituting subsystems.

There are different approaches to multipartite entanglement quantification. A proposed measure, introduced in Ref.~\cite{meyer02} by Meyer and Wallach, quantifies the global entanglement of the multipartite system, and vanishes for fully separable states only. Another approach, introduced by Love et. al. in Ref.~\cite{love07}, quantifies the amount of genuine multipartite entanglement and vanishes for any partially separable states. 
A different approach defines families of entanglement measures that quantifies the amount of entanglement also for intermediate or partial separability, as proposed in Refs.~\cite{hiesmayr08, hiesmayr09, ichikawa09}. In Ref.~\cite{ichikawa09}, a family of entanglement measures for intermediate separability of pure states of $n$ qubits, called $R_m$ measures, has been introduced. This family includes the Meyer-Wallach measure and the Love measure as elements of the family. It manifests its usefulness by exhibiting a clear difference between the well-known multipartite GHZ and W states for systems of up to fifty qubits. To generalize the entanglement measures to mixed states, one uses the well-established convex roof extension. The measure is then defined as the weighted measure for pure states of the mixed states' decomposition, where one has to take the infimum over all possible decompositions. However, it is in general hard to calculate the convex roof, since mixed states allow infinitely many decompositions into pure states.

The aim of this paper is to derive lower bounds for the $R_m$ measures for mixed states of $n$ qudits. Lower bounds guarantee at least a certain value of entanglement that has to be present in the system. To do so, we utilize a method introduced in Ref.~\cite{hiesmayr08} and similarly in Ref.~\cite{mintert05b}, and applied in Refs.~\cite{hiesmayr08a, hiesmayr09}. This method decomposes entanglement measures into sums of generalized concurrences (so-called $M$-concurrences), where for each $M$-concurrence a lower bound can be easily computed, and thus it can also be achieved for the entanglement measure. In the following we give the necessary definitions, show how to decompose the $R_m$ measures into $M$-concurrences, and thus are able to derive a formula for a lower bound of the $R_m$ measures for mixed states. We illustrate these results by instructive examples of four-partite states.


\section{Measures and their Lower Bounds}

We consider an $n$-qudit system ${\cal H}=\Bot_{i=1}^n{\cal H}_i$ with constituent systems ${\cal H}_i=\C^d$ for all $i$. To specify how to focus on the total system, let us introduce the partition set $\Gamma:=\{\gamma_j\}_{j=1}^m$, whose elements satisfy
\be
\bigcup_{j=1}^m\gamma_j={\cal N},
\quad
{\rm and}
\quad
\gamma_j\cap\gamma_k=\emptyset
\quad
{\rm for}
\quad
j\neq k,
\ee
where ${\cal N}:=\{1,2,\cdots,n\}$ is the set of the labels of the constituents, and $m$ is the total number of subsystems. We denote the complement of $\gamma_j$ with regard to ${\cal N}$ by $\bar{\gamma}_j$ and the number of the elements of the (sub)set $\gamma$ by $|\gamma|$ (see FIG. 1).

\begin{figure}[t]
\begin{center}
\includegraphics[width=2in]{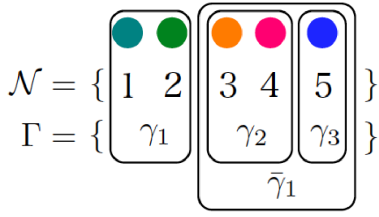}
\caption{Schematic diagram for the partition with $n=|{\cal N}|=5, m=3$, $|\gamma_1|=|\gamma_2|=2$ and $|\gamma_3|=1$.}
\label{ }
\end{center}
\end{figure}

Let us start by defining generalized concurrences for multipartite states (for bipartite pure qudit states, related generalized concurrences were introduced in \cite{audenaert01, rungta01}). For that purpose, let us introduce flip operators $\sigma_{kl}^i$ with $k,l=0,1,\cdots,d-1$ that act on the $i$-th qudit and are defined as
\be
    &\sigma_{kl}^i := \ket{k} \bra{l} + \ket{l} \bra{k} \,.
\ee
Flip operators $f_\delta(\{k_i,l_i\}_{i\in\delta})$ with respect to the set of subsystems $\delta\subseteq{\cal N}$ are given by
\be
f_\delta(\{k_i,l_i\}_{i\in\delta}):=\Bot_{i\in\delta}\sigma^i_{k_il_i}\ot\Bot_{j\in\bar{\delta}}\mathbbm{1}^j,
\ee
where $\mathbbm{1}^i$ is the identity operator on ${\cal H}_i$. Since it is evident that the indices in the argument of $f_\delta$ run only within the elements of $\delta$, hereafter we use the abbreviated notation $f_\delta(\{k_i,l_i\})$ for the flip operators. Using the notation $\ket{ \{ j \} }:= \ket{j_1} \otimes \ket{j_2} \otimes \ldots \otimes \ket{j_n}$ with $\braket{k|l}=\delta_{kl}$ and $k,l=0,1,\cdots,d-1$ for states of the computational basis, we construct an operator
\be
\begin{split}
O_{\gamma,\delta}(&\{k_i,l_i\};\{j\}):=f_\delta(\{k_i,l_i\})\ket{\{j\}}\bra{\{j\}}\\
&-f_{\gamma}(\{k_i,l_i\})\ket{\{j\}}\bra{\{j\}}f_{\bar{\gamma}\cap\delta}(\{k_i,l_i\})
\end{split}
\ee
by the help of two subsets $\gamma, \delta$ of ${\cal N}$ satisfying $\delta \subseteq {\cal N}$ and $\delta \cap \gamma = \emptyset$. In the following we use the abbreviated notation $O_{\gamma,\delta}$ if there is no possibility to cause confusions. The generalized (squared) concurrences (also called $M$-concurrences) of pure states $\ket{\psi}\in{\cal H}$ for the two subsets $\gamma, \delta$ are defined as
\be
\begin{split}
C_{\gamma,\delta}^2(\psi)
:=\sum_{i\in\delta}\sum_{k_i<l_i}\sum_{\ket{ \{ j \} }} &|\bra{\psi}O_{\gamma,\delta}\ket{\psi^*}|^2,
\end{split}
\label{defCpure}
\ee
which vanishes if and only if the state $\ket{\psi}$ is separable with respect to the bipartition $\{\gamma,\bar{\gamma}\}$. Here, $\ket{\psi^*}$ denotes the complex conjugated state to $\ket{\psi}$.
Note that the Hill-Wootters concurrence \cite{hill97} for pure states is reproduced for $n=2$ and $d=2$ (two qubits).

We define the generalized (squared) concurrences for mixed states via the convex roof,
\be
C^2_{\gamma,\delta}(\rho) := \inf_{\{p_\alpha,\psi_\alpha\}}\sum_\alpha p_\alpha C^2_{\gamma,\delta}(\psi_\alpha) \,,
\label{defCmix}
\ee
where the infimum is taken over all possible decompositions of the given density matrix $\rho=\sum_\alpha p_\alpha\ket{\psi_\alpha}\bra{\psi_\alpha}$ into a probability distribution $\{p_\alpha\}$ and pure states $\ket{\psi_\alpha}$. Although it is in general hard to evaluate the convex roof extension, we can explicitly determine a lower bound for the generalized concurrences by
\be
C^2_{\gamma,\delta}(\rho) \geq \Lambda^2_{\gamma,\delta}(\rho)
\label{Clb}
\ee
with
\be
\Lambda_{\gamma,\delta}(\rho):=\max\left\{0,\sum_{O_{\gamma,\delta}}\(2\lambda(O_{\gamma,\delta})-\Tr\sqrt{\rho\tilde{\rho}(O_{\gamma,\delta})}\)\right\},
\ee
where the summation is taken under the same condition as in Eq.~(\ref{defCpure}) and $\lambda(O_{\gamma,\delta})$ is the largest eigenvalue of $\sqrt{\rho\tilde{\rho}(O_{\gamma,\delta})}$ with
\be
\tilde{\rho}(O_{\gamma,\delta}):=(O_{\gamma,\delta}+O_{\gamma,\delta}^\dag)\rho^*(O_{\gamma,\delta}+O_{\gamma,\delta}^\dag).
\ee
This lower bound will be helpful in the following derivation of the lower bound for the $R_m$ measures.

Summing up all generalized concurrences for subsets $\delta$, we can define an entanglement measure for the set $\gamma$ of constituents,
\be 
\label{etagamma}
\eta_{\gamma}(\psi):=\sum_{\delta\subseteq{\cal N}} C_{\gamma,\delta}^2(\psi),
\label{defetapure}
\ee
where the sum is restricted by $\delta \cap \gamma = \emptyset$. The relation to the linear entropy of the reduced density matrix $\rho_{\gamma}:=\Tr_{\bar{\gamma}}\ket{\psi}\bra{\psi}$ of the set $\gamma$ of the constituents has been established in \cite{hiesmayr08, hiesmayr09}:
\be
\eta_{\gamma}(\psi)=N(|\gamma|)\(1-\Tr\rho_{\gamma}^2\),
\ee
where  $N(|\gamma|):=d^{|\gamma|}/(d^{|\gamma|}-1)$ is a normalization factor in order to obtain $\eta_{\gamma} (\psi) = 1$ if $\rho_{\gamma}$ is the maximally mixed state.

Let us generalize the measure $\eta_{\gamma}(\psi)$ \eqref{etagamma} to measures for particular partitions of the $n$-qudit system. To do so, we rewrite $\gamma\subseteq{\cal N}$ as $\gamma_i$, such that it can be regarded as an element of general partition $\Gamma=\{\gamma_i\}_{i=1}^m$. Taking the arithmetic average of $\eta_{\gamma_i}(\psi)$ for all $i$, we define the following entanglement measures with regard to a particular partition ${\Gamma}$:
\be
\xi_\Gamma(\psi):=\frac{1}{m}\sum_{i=1}^m\eta_{\gamma_i}(\psi).
\label{defxipure}
\ee
Furthermore, taking the geometric average of $\xi_{\Gamma}$ for all possible partitions under the condition that the number of the elements of the partitions $m$ is fixed, we obtain the family of entanglement measures $R_m(\psi)$ for intermediate separability,
\be
R_m(\psi):=\(\prod_{|\Gamma|=m}\xi_\Gamma(\psi)\)^{1/S(n,m)},
\ee
where
\be
S(n,m):=\sum_{k=1}^m\frac{(-)^{m-k}k^{n-1}}{(k-1)!(m-k)!}
\ee
is the Stirling number in the second kind, representing the number of subsystem combinations that result in $m$ partitions.

In order to generalize $R_m(\psi)$ to mixed states, we take the convex roof of $\xi_\gamma(\psi)$,
\be
\xi_\Gamma(\rho):=\inf_{\{p_\alpha,\psi_\alpha\}}\sum_\alpha p_\alpha\xi_\Gamma(\psi_\alpha),
\label{defximix}
\ee
and define $R_m(\rho)$ for mixed states as a quantity obtained by taking the geometric average of $\xi_\Gamma(\rho)$ for all possible partitions with the fixed number of subsystems.

For that purpose, we have to prove the following lemma:
\begin{lemma} \label{thmroof}
Suppose that there exist pure state entanglement measures $\mu_s(\psi)$, labeled by the index $s$. Then, the convex roof of the sum of them is no less than the sum of the convex roofs of each measure, i.e.
\be
\mu(\rho)\ge\sum_s\mu_s(\rho)
\ee
where
\be
\mu(\rho):=\inf_{\{p_\alpha,\psi_\alpha\}}\sum_\alpha p_\alpha\sum_s\mu_s(\psi_\alpha)
\label{defmu}
\ee
and
\be
\mu_s(\rho):=\inf_{\{p_\alpha,\psi_\alpha\}}\sum_\alpha p_\alpha\mu_s(\psi_\alpha).
\ee
\end{lemma}
\begin{proof}
Suppose that the decomposition of the given mixed state $\rho$ which yields $\mu(\rho)$ is given by $\rho=\sum_\alpha p_\alpha^\prime\ket{\psi_\alpha^\prime}\bra{\psi_\alpha^\prime}$. Then, starting from Eq.~(\ref{defmu}), we get
\be
\begin{split}
\mu(\rho)
=\sum_{s,\alpha}p_\alpha^\prime\mu_s(\psi_\alpha^\prime)
=\sum_s\sum_\alpha p_\alpha^\prime\mu(\psi_\alpha^\prime)
\ge\sum_{s}\mu_s(\rho)
\end{split}
\ee
where we have to use ``$\geq$'' since the decomposition $\{p_\alpha^\prime, \ket{\psi_\alpha^\prime}\}$ does not necessarily yield the infimum of $\sum_\alpha p_\alpha\mu_s(\psi_\alpha)$ for all $s$.
\end{proof}

Applying Lemma~\ref{thmroof} to Eq.~(\ref{defximix}), we  obtain a lower bound of $\xi_\Gamma(\rho)$:
\be
\xi_\Gamma(\rho)\ge\frac{1}{m}\sum_{i=1}^m\eta_{\gamma_i}(\rho),
\ee
where
\be
\eta_{\gamma_i}(\rho):=\inf_{\{p_{\alpha},\psi_{\alpha}\}}\sum_{\alpha} p_{\alpha} \eta_{\gamma_i}(\psi_{\alpha}).
\label{defetamix}
\ee
Furthermore, utilizing (\ref{defetapure}), Lemma~\ref{thmroof}, (\ref{defCmix}), and Ineq.~(\ref{Clb}) for Eq.~(\ref{defetamix}), we find
\be
\begin{split}
\eta_{\gamma_i}(\rho)&=\inf_{\{p_{\alpha},\psi_{\alpha}\}}\sum_{\alpha} p_{\alpha} \eta_{\gamma_i}(\psi_{\alpha})\\
&=\inf_{\{p_{\alpha},\psi_{\alpha}\}}\sum_{\alpha} p_{\alpha}\sum_{\delta}C_{\gamma_i,\delta}^2(\psi_\alpha)\\
&\ge\sum_{\delta}C_{\gamma_i,\delta}^2(\rho)
\geq\sum_{\delta}\Lambda_{\gamma_i,\delta}^2(\rho).
\end{split}
\ee
Thus, a computable lower bound of $R_m(\rho)$ is given by
\be
R_m(\rho)\geq \tilde{R}_m(\rho) \,,
\ee
where
\be
\tilde{R}_m(\rho):=\frac{1}{m}\(\prod_{|\Gamma|=m}\sum_{i=1}^m\sum_{\delta\subseteq{\cal N}} \Lambda_{\gamma_i,\delta}^2(\rho)\)^{1/S(n,m)},
\ee
and the second sum is again conditioned by $\delta \cap \gamma_i = \emptyset$ for each $i$.

\section{Example}

\begin{table}[t]
  \begin{center}
  \begin{tabular}{cc}
\hline
    Representative partition & Equivalent partitions\\
\hline
  $\{\{1\},\{2,3,4\}\}$   & $\{\{2\},\{1,3,4\}\}$ \\
  $\{\{3\},\{1,2,4\}\}$   & $\{\{4\},\{1,2,3\}\}$ \\
  $\{\{1,3\},\{2,4\}\}$   & $\{\{1,4\},\{2,3\}\}$ \\
  $\{\{1,2\},\{3,4\}\}$   &  \\
\hline
\end{tabular}
\end{center}
  \centering
  \begin{tabular}{cc}
\hline
    Representative partition & Equivalent partitions\\
\hline
  $\{\{1\},\{2\},\{3,4\}\}$   & \\
  $\{\{3\},\{4\},\{1,2\}\}$   & \\
  $\{\{1\},\{3\},\{2,4\}\}$   & $\{\{1\},\{4\},\{2,3\}\}$ \\
  &$\{\{2\},\{3\},\{1,4\}\}$\\
  &$\{\{2\},\{4\},\{1,3\}\}$\\
\hline
\end{tabular}
  \caption{Classification of the partitions of four-partite systems. The equivalent partitions can be mapped into the representative partition in the same line by the actions of the elements of $V$. (Above) The classification of the bipartitions. (Below) The classification of the tripartitions.}
  \label{ }
\end{table}

As an example of an explicit calculation of the lower bound formula, let us consider a family of four-qubit mixed states on $(\C^{2})^{\ot4}$
\be
\rho=p_1P^+_{12}\otimes P^+_{34}+p_2P^{\rm GHZ}+\frac{1-p_1-p_2}{16}\Bot_{i=1}^4\mathbbm{1}^i.
\label{BBO}
\ee
Here, $P^+_{ij}$ is that onto $\ket{\phi^+}_{ij}$, one of Bell bases spanning ${\cal H}_i\ot{\cal H}_j$, that is,
$
\ket{\phi^+}_{ij}:=\(\ket{00}_{ij}+\ket{11}_{ij}\)/\sqrt{2},
$
$P^{\rm GHZ}$ is the projector onto the GHZ state
$
\ket{\rm GHZ}:=\(\ket{0000}+\ket{1111}\)/\sqrt{2},
$
and
$
0\le p_1+p_2\le1,
$
with
$
p_1, p_2\ge0.
$
Note that the first and second terms in Eq. (\ref{BBO}) can be produced from the second order non-linear effect of a $\beta-{\rm BaB_3O_6}$ (BBO) crystal \cite{pan98, pan01}, respectively. Since the third term can be regarded as the isotropic noise, we expect that due to the quantification of entanglement of the state, we can see not only a variety of entanglement produced by the BBO crystal, but also how much the noise affects the entanglement in the system.
\begin{figure*}[t]
\begin{center}
\includegraphics[width=2in]{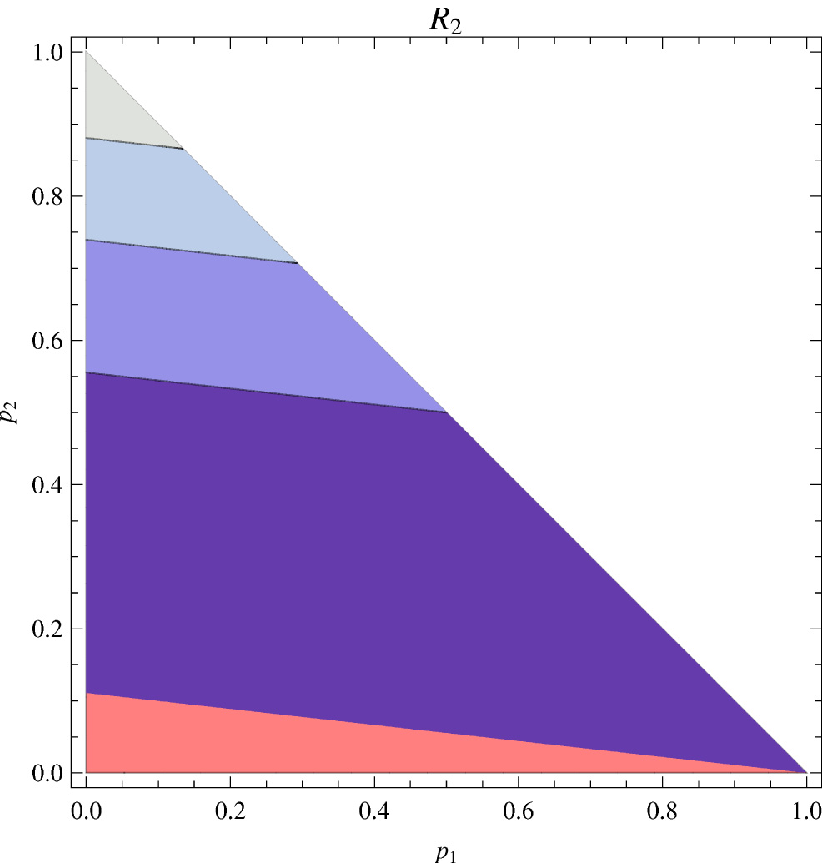}
\qquad
\includegraphics[width=2in]{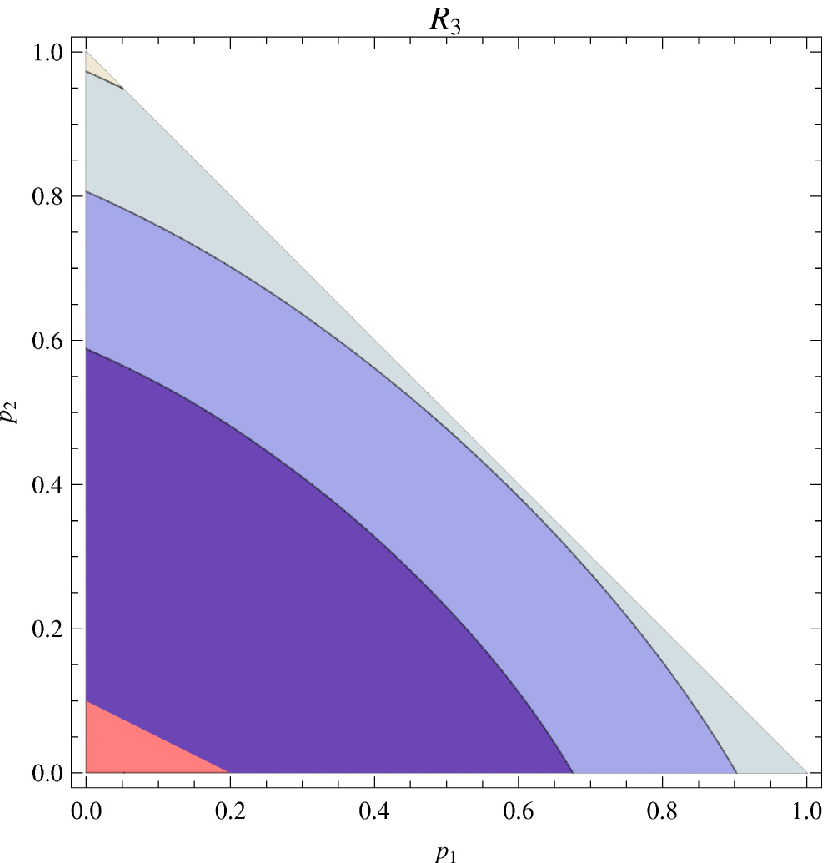}
\qquad
\includegraphics[width=2in]{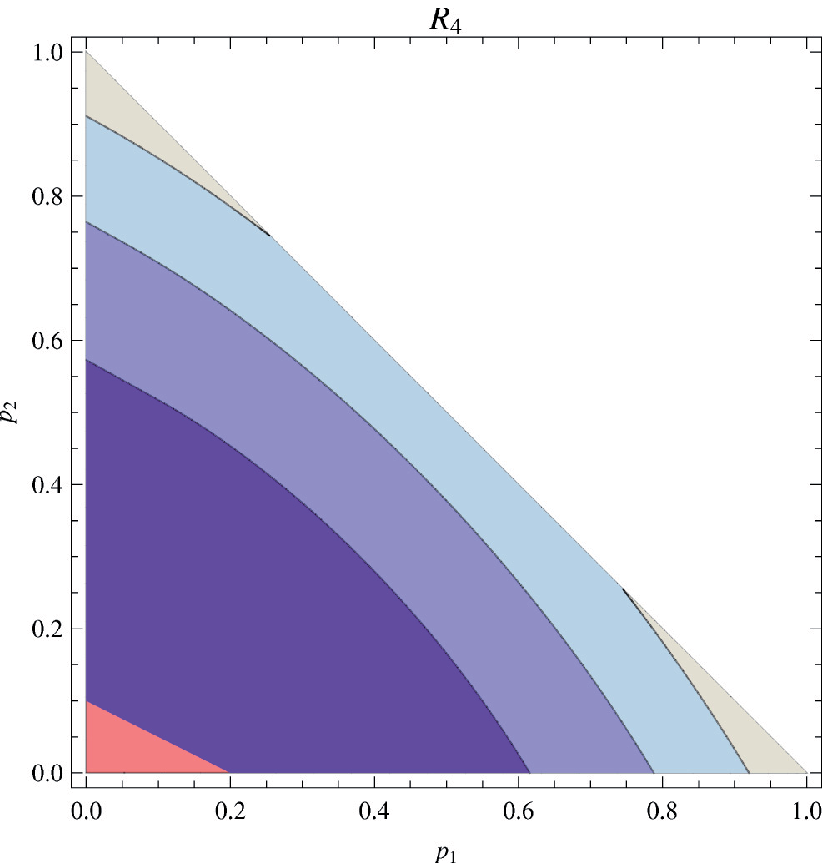}
\setlength{\unitlength}{1cm}%
\caption{Contour plots of the values of the lower bounds $\tilde{R}_m$ of $R_m$ measures for the state (\ref{BBO}). (Left) $R_2$ measure. (Center) $R_3$ measure. (Right) $R_4$ measure. Each colored area denotes the region where the lower bound of the measure has the specific range (red: $\tilde{R}_m=0$, dark purple: $0<\tilde{R}_m\le0.25$, bright purple: $0.25<\tilde{R}_m\le0.5$, blue: $0.5<\tilde{R}_m\le0.75$, ash: $0.75<\tilde{R}_m\le1$).}
\label{Rmvalues}
\end{center}
\end{figure*}

The state in Eq. (\ref{BBO}) clearly lacks the symmetry under actions of $S_4$ for the labels of the constituents, which is due to the first term in the summand. However, we easily see that it still holds a symmetry under actions of four elements of $S_4$,  i.e. the identity operation $e$, two transpositions $(1,2)$, $(3,4)$ and their consecutive operation $(1,2)(3,4)$. We can find that these four elements constitute a subgroup of $S_4$, which has been known as Vierergruppe $V$ \cite{armstrong88}. Hence, $\rho$ is invariant under the actions of $V$.
Such symmetry has a relevant role for the reduction of computational complexity. For example, the number of the bipartitions of four-partite systems, $S(4,2)=7$, effectively reduces to $4$. By the same way, that of tripartitions, $S(4,3)=6$, reduces to $3$ (see TABLE I).

The amount of entanglement in the state (\ref{BBO}) is visualized in FIG. \ref{Rmvalues}. Notice that the area with $R_m=0$ with larger $m$ is included in the same area with smaller $m$. This reflects the fact that the lower bound $\tilde{R}_m(\rho)$ captures the property that a $m$-separable state can be regarded as a $m^\prime$-separable state with $m\ge m^\prime$. To analyze these graphs in more detail, it is convenient to introduce two variables
\be
q:=1-p_1-p_2
\quad
{\rm and}
\quad
r:=\frac{p_2}{p_1}.
\ee
The former variable $q$ corresponds to the degree of the noise, while the latter $r$ characterizes the original noiseless state which has been altered into the state specified by the coordinates $(p_1,p_2)$ due to the presence of noise.

Keeping $q$ fixed and varying $r$, let us observe the variety of entanglement under the fixed noise situation. We can immediately see that the $R_2$ measure decreases monotonically as $r$ decreases, while the others behave differently. Since the smaller value of $r$ implies that the ratio of the bi-separable state $P_{12}^+\ot P_{34}^+$ in $\rho$ becomes larger, the preceding observation means that by the addition of the biseparable state, the state approaches the biseparable state monotonically, while the state does not approach the tri-separable or four-separable state. This comes from the fact that $P_{12}^+\ot P_{34}^+$ in $\rho$ is a genuinely biseparable state, and cannot be regarded as a tri-separable or four-separable state.
On the other hand, varying $q$ and fixing $r$, we see that all graphs share a common behavior: the monotonic approach to $R_m = 0$ for all $m$ by the addition of the noise. This is due to the fact that the noise $\(\Bot_i\mathbbm{1}^i\)/16$ can be interpreted as a separable state for any partition. From these observations, we may conclude that the lower bound $\tilde{R}_m(\rho)$ derived in this letter captures the natural behavior of the multipartite entanglement suitably.

\section{Summary}
In this letter, starting from the $m$-concurrences, we  systematically derived the computable lower bound of the family of the entanglement measures $\{R_m(\rho)\}_{m=2}^n$ by utilizing Lemma 1, which manifests the non-commutativity of the convex roof extention and summations of entanglement measures. As a testing ground of the derived lower bound, we examined the amount of the entanglement of the state (\ref{BBO}) and showed that the resultant graphs are explained by the natural behavior of the system in question. Thus, this example confirms the consistency of the lower bound and is useful for a finer analysis of entanglement.

\noindent Acknowledgements:
Tsubasa Ichikawa is supported by \lq Open Research Center\rq~Project for Private Universities: matching fund subsidy from MEXT, Japan. Marcus Huber acknowledges the Austrian Science Fund project FWF-P21947N16. Philipp Krammer acknowledges financial support by FWF project CoQuS No. W1210-N16 of the Austrian Science Fund. 

\bibliography{refsrmbound}

\end{document}